\documentclass[preprint,review,12pt]{elsarticle}

\usepackage{color}
\usepackage{hyperref}
\usepackage{amsmath,amssymb,graphicx}
\usepackage{amsfonts,amsthm}
\usepackage{multirow}
\usepackage{color}

\newtheorem{proposition}{Proposition}

\newtheorem{remark}{Remark}
\newtheorem{model}{Model}

\begin{document}
%\linenumbers
\begin{frontmatter}

\title{Handling missing data in large healthcare dataset:\\ a case study of unknown trauma outcomes}

\author[LeicMath]{E.M. Mirkes}
\ead{em322@le.ac.uk}
\author[LeicMed]{T.J. Coats}
\ead{tc61@le.ac.uk}
\author[LeicMath]{J. Levesley}
\ead{jl1@le.ac.uk}
\author[LeicMath]{A.N. Gorban\corref{cor1}}
\ead{ag153@le.ac.uk}

\address[LeicMath]{Department of Mathematics, University of Leicester, Leicester, LE1 7RH, UK}
\address[LeicMed]{Emergency Medicine Academic Group, Department of Cardiovascular Sciences, University of Leicester, Leicester, LE1 7RH, UK}
\cortext[cor1]{Corresponding author}

\begin{abstract}
Handling of missed data is one of the main tasks in data preprocessing especially in
large public service datasets. We have analysed data from the Trauma Audit and Research
Network (TARN) database, the largest trauma database in Europe. For the analysis we used
165,559 trauma cases. Among them, there are 19,289 cases (13.19\%) with unknown outcome.
We have demonstrated that  these outcomes are not missed `completely at random' and, hence, it
is impossible just to exclude these cases from analysis despite the large amount of
available data. We have developed a system of non-stationary Markov models for the handling of
missed outcomes and validated these models on the data of 15,437 patients which
arrived into TARN hospitals later than 24 hours but within 30 days from injury. We used
these Markov models for the analysis of mortality. In particular, we corrected the observed
fraction of death. Two na\"ive approaches give 7.20\% (available case study) or 6.36\%
(if we assume that all unknown outcomes are `alive'). The corrected value is 6.78\%.
Following the seminal paper of Trunkey (1983) the multimodality of mortality curves has become 
a much discussed idea. For the whole analysed TARN dataset the coefficient of mortality
monotonically decreases in time but the stratified analysis of the mortality gives a
different result: for lower severities the coefficient of mortality is a non-monotonic
function of the time after injury and may have maxima at the second and third weeks. The
approach developed here can be applied to various healthcare datasets which experience the
problem of lost patients and missed outcomes.
\end{abstract}

\end{frontmatter}

\section{Introduction}\label{introd}

{Enthusiasm for the use of big data in the improvement of health service is huge but there is
a concern that without proper attention to some specific challenges the mountain of 
big data efforts} will bring forth a mouse \cite{Adler2013}. Now, there is no technical
problem with ``big'' in healthcare. Electronic health records include hundreds of
millions of outpatient visits and tens of millions {of} hospitalizations, and these numbers
grow exponentially. The main problem is in quality of data.

``Big data'' {very often means} ``dirty data'' and the fraction of {\em data
inaccuracies} increases with data volume growth. {Human} inspection at the big data scale
is impossible and there is a desperate need {for} intelligent tools for accuracy and
believability control.

The second big challenge of big data in healthcare is {\em missed information}. There may
be many reasons for data incompleteness. One of them is {in health} service
``fragmentation''. This problem can be solved partially by the national and international
unification of the electronic health records (see, for example, Health Level Seven
International (HL7) standards \cite{Dolin2006} or discussion of the template for uniform
reporting {of} trauma data  \cite{TARN1_2008}). However, some fragmentation is
unavoidable due to the diverse structure of the health service. In particular, the modern
tendency for personalization of medicine can lead to highly individualized sets of
attributes for different patients or patient groups. There are several universal
technologies for the handling of missing data
\cite{Rubin1987,Rubin1996,Pigott2001,Schafer2002,Graham2007,Graham2012,Donders2006}.
Nevertheless, the problem of handling missed values in large healthcare datasets  is
certainly not completely solved. It continues to attract the efforts of many researchers
(see, for example, \cite{Cismondia2013}) because the popular universal tools can lead to
bias or loss of statistical power \cite{Gorelick2006,Sterne2009}. For each system, it
is desirable to combine various {existing} approaches for the handling of missing data (or to
invent new ones) to minimize the damage to the results of data analysis. For the best
possible solution, we have to take into account the peculiarities of each database and to
specify the further use of the cleaned data (it is desirable to understand in advance how
we will use the preprocessed data).

In our work we analyze missed values in {the} TARN database \cite{TARNweb}. We use the
preprocessed data for:
\begin{itemize}
\item  the evaluation of the risk of death,
\item the identification of the patterns of mortality,
\item approaching several old problems like the Trunkey hypothesis about the
    trimodal distribution of trauma mortality \cite{Trunkey1983}.
\end{itemize}

The `two stage lottery' non-stationary Markov model developed in the sequel can be used for
the analysis of missing outcomes in a much wider context than the TARN database and could
be applied to the handling of data gaps in healthcare datasets which experience the
problem of transferred and lost patients and missing outcomes.

In this paper we analyze the unknown outcomes. The next task will be the analysis of missed
data in the most common ``input'' attributes.

\section{Data set}

There are more than 200 hospitals which send information to TARN (TARN hospitals). This
network is gradually increasing.  Participation in TARN is recommended by the Royal
College of Surgeons of England and the Department of Health. More than 93\% of hospitals
across England and Wales submit their data to TARN. TARN also receives data from Dublin,
Waterford (Eire), Copenhagen, and Bern.

We use TARN data collected from 01.01.2008 (start of treatment) to 05.05.2014 (date of
discharge). The database contains 192,623 records and more than 200 attributes. Sometimes
several records correspond to the same trauma case because the patients may be
transferred between TARN hospitals. We join these records. The resulting database
includes data of 182,252 different trauma cases with various injuries.

\begin{figure*}
\begin{centering}
\includegraphics[width= 0.8\textwidth]{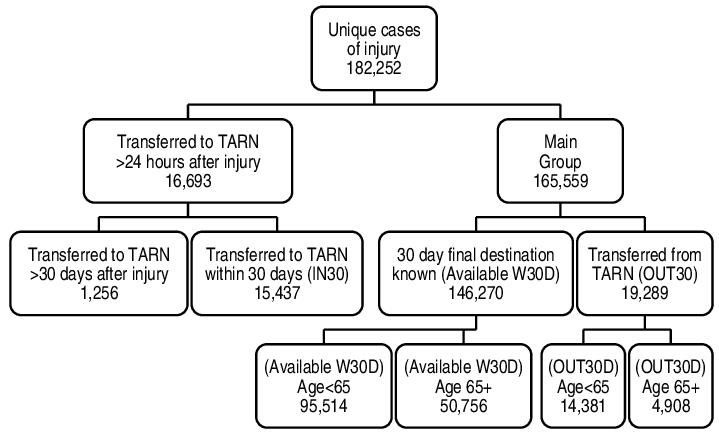}
\caption{The groups of the patients for analysis of mortality. FOD in the group
`Available W30D' can be calculated from the data directly. Mortality in the group `OUT30'
will be evaluated on the basis of the non-stationary Markov model. The group of 16,693
patients which arrived (were transferred from other institutions) to TARN hospitals later
than 24 hours after injury was excluded from the mortality analysis. Its subgroup `IN30'
of 15,437 patients is used for validation of the Markov model for `OUT30' group. The subgroups  with age$<65$ and age$\geq 65$ should be separated because for age$\geq 65$ the following traumas 
are excluded from the database: Acetabulum fractures (AIS 8562xx), Pelvic/Acetabulum fractures (AIS 8563xx), 
Pelvic ring fractures (AIS 8561xx), Pubic rami and Femoral neck fractures (AIS 85316x). 
\label{MainGroups}}
\end{centering}
\end{figure*}

16,693 records correspond to patients, who arrived (transferred from other institutions)
to TARN hospitals later than 24 hours after injury. This sample is biased, for example
the Fraction Of Dead outcomes (FOD) for this sample is 3.34\% and FOD for all data is
6.05\%. This difference is very significant for such a big sample. (If all the outcomes
in a group of the trauma cases are known then we use the simple definition of FOD in the
group: the ratio of the number of registered deaths in this group to the total number of
patients there. Such a definition is not always applicable. The detailed and more
sophisticated analysis of this notion follows in the next {section}.) We remove these
16,693 trauma cases from analysis but use them later for validation of the ``mortality
after transfer'' model. Among them, there are 15,437 patients who arrived at a TARN
hospital within 30 days after injury. We call this group `IN30' for short
(Fig.~\ref{MainGroups}).

As a result we have 165,559 records for analysis (`Main group'). This main group consists
of two subgroups: 146,270 patients from this group approached TARN during the first day
of injury and remained in TARN hospitals or discharged to a final destination during the
first 30 days after injury. We call this group the `Available within 30 days after
injury' cases (or `Available W30D' for short). The other 19,289 patients have been
transferred  within 30 days after injury to a hospital or institution (or unknown
destination) who did not return data to the TARN system. We call them `Transferred OUT OF TARN within 30 days after injury'  or just `OUT30' (Fig.~\ref{MainGroups}).

The patients with the non-final discharge destinations `Other Acute hospital' and `Other
institution' were transferred from a TARN hospital to a hospital (institution) outside
TARN and did not return to the TARN hospitals {within} 30 days after injury.

The database includes several indicators for evaluation of the severity of the trauma
case, in particular, Abbreviated Injury Scale (AIS), Injury Severity Score (ISS) and New
Injury Severity Score (NISS). For a detailed description and comparison of the scores
we refer readers to reviews \cite{Lefering2002,Lesky2014}. The comparative study of {predictive}
ability of different scores has a long history
\cite{Goldfarb1977,Sacco1988,Champion1996,Rutledge1998}. The scores are used for
mortality predictions and {are} tested on different datasets
\cite{Sullivan2003,Lavoie2004,Tay2004,TARN8_2006}.
{In the database, there exist no  gaps in AIS (and hence ISS and NISS) values even for patients rapidly dying. Most severely injured patients have a CT `pan-scan' within the first hour or two of injury which is likely to define all life-threatening injuries. In addition the report from the post-mortem examination is used in the compilation of an injuries list which is the basis of AIS, and hence ISS and NISS, scoring.}

\section{Definitions and distributions of outcomes}

The widely used definition of the endpoint outcome in trauma research is survival or
death within 30 days after injury \cite{Clark2004,TARN8_2006,Skaga2008}.

A substantial number of TARN in-hospital deaths following trauma {occur} after 30 days:
there are 957 such cases (or 8\% of TARN in-hospital death) among 11,900 cases with
`Mortuary' discharge destination. This {proportion} is practically the same in the
main group (165,559  cases): 894 deaths after 30 days in hospital (or
7.9\%) among 11,347 cases with `Mortuary' discharge destination.

{Death} later than 30 days after injury may be considered as caused by co-morbidity
rather than the direct consequence of the injury \cite{TARN8_2006}. These later deaths
are not very interesting from the perspective of an acute trauma care system (as we
cannot influence them), but they might be very interesting from the perspective of a
geriatric rehabilitation centre or of an injury prevention program for elderly patients.

On the other hand, when ``end of acute care'' is used as an outcome definition then a
significant portion of deaths remains unnoticed. For example, in the 3332 trauma cases
treated in the Ulleval University Hospital (Oslo, Norway, 2000-2004) 18\% of deaths
occurred after discharge from the hospital \cite{Skaga2008}.

The question of whether it is possible to {neglect trauma} caused mortality within 30 days
after trauma for the patients with the discharge destination `Home', `Rehabilitation' and
other `recovery' outcomes is not trivial \cite{Skaga2008}. Moreover, here are two
questions:
\begin{itemize}
\item How {do we} collect all the necessary data after discharge within 30 days after
    trauma -- a technical question?
\item How {do we} classify the death cases after discharge within 30 days after trauma;
    are they consequences of the trauma or should they be considered as comorbidity with some additional reasons?
\end{itemize}
The best possible answer to the first question requires the special combination of
technical and business process to integrate data from different sources. The recent
linkage from TARN to the Office for National Statistics (ONS) gives the possibility to access
the information about the dates of death in many cases. It is expected that the further
data integration process will recover many gaps in the outcome data.

The last question is far beyond the scope of data management and analysis and may be
approached from different perspectives. Whether or not the late deaths are important in a
model depends on the question being asked. From the data management perspective, we have
to give the formal definition of the outcome in the terms of the available database
fields. It is impossible to use the standard definition as survival or death within 30
days after injury because these data are absent. We define the outcome `Alive W30D' for
the TARN database being as close to the standard definition as it is possible.

In the TARN database discharge destinations `Home (own)', `Home (relative or other carer)',
`Nursing Home', and `Rehabilitation' are considered as final.  If we assume that these
trauma cases have the outcome `Alive  W30D' then we loose some cases of death. From the
acute care perspective these cases can be considered as irrelevant. Let us accept this
definition. There still remain many cases with unknown outcome. For analysis of these
cases we introduce the outcome category `Transferred'. In this category we include the
cases which left the TARN registry to a hospital or other institution outside TARN, or to
an unknown destination within 30 days. The relations between the discharge destinations
and these three outcomes are presented in Table~\ref{Table:2}.

\begin{table*}{\footnotesize\begin{center}
\caption{\label{Table:2} Distribution of outcomes in the main group (W30D means within 30
days after injury).}
\begin{tabular}{ccccc}
\hline Subgroup  & Alive W30D & Dead  W30D  & Unknown  & Total \\
\hline Available W30D & 135,733 &10,537 &0 & 146,270 \\
\hline OUT30 & 0$^*$ & 0$^*$ & 19,279 & 19,289 \\
\hline Total & 135,733& 10,537 & 19,289 & 165,559 \\    \hline
\end{tabular}
\\ $^*$No known survival or deaths.
\end{center}}
\end{table*}

As we can see from Table~\ref{Table:2}, 19,289  trauma cases (or 11.35\% of all cases)
have unknown outcome. The first standard question is: can we delete these data and {apply
{\em available case analysis}}? For this purpose we have to consider these outcome data
as ``Missing Completely at Random'' (MCAR)
\cite{Rubin1976,Rubin1987,Pigott2001,Schafer2002}. This is definitely not the case. The
group with unknown outcomes is exactly the `OUT30' group. The probability of belonging
to this group depends, for example, on the severity of injury (which can be measured, by the
maximal severity, by NISS, by GCS or by another severity score). The $\chi^2$ test of
independence shows that transfer depends on the severity with $p$-value $p<10^{-300}$
({this is the probability that} such a strong dependence might appear by chance).  {The 
most practical (or `purposefull' \cite{Fuchs2013}) idea is to consider the missed outcome data as ``Missed at Random'' (MAR).
The assumption of MAR does not imply that the data are
missing randomly, but rather that the missing values are correlated with variables recorded
in the dataset \cite{Fuchs2013}}.

One can consider all these cases as alive because these patients have been alive at the
{point of} discharge from TARN hospitals. If we consider all transferred as alive then
the FOD is 6.35\%. If we delete all the transferred patients (study only the Available
W30D group) then the FOD is 7.2\%. If we test this hypothesis on 15,437 patients of the
group `IN30' transferred {\em to} TARN hospitals from outside the network within 30 days
after injury then we find that the nonzero mortality {for} them (3.10\%).

The data table with known outcomes is necessary for further machine learning and the
main goal is outcome prediction and risk evaluation. 

We choose to remove the OUT30D group from data table  but simultaneously
to adjust the weights of the retained cases to compensate for the removal. The
information about the OUT30D cases will be used in the construction of the
weights. It is necessary to evaluate the mortality of the patients transferred from TARN
before removing their records and reweighting of the rest. In the next section we
develop, identify and validate Markov models for the analysis of the mortality of
transferred patients.

Another method for handling missed outcomes  is multiple imputation of 
the outcomes (about multiple imputations see, for example,
\cite{Graham2012}). Both methods use similar stochastic {models} of mortality and
transfer. The large number of cases allows us to use the reweighting approach. A
significant majority of the evaluated weights are between 0.9 and 1.1 (see
Section~\ref{weighting}).

\section{Non-stationary Markov model for the analysis of missing outcomes}

\subsection{Structure of model}

\begin{figure}
\begin{centering}
\includegraphics[width= 0.35\textwidth]{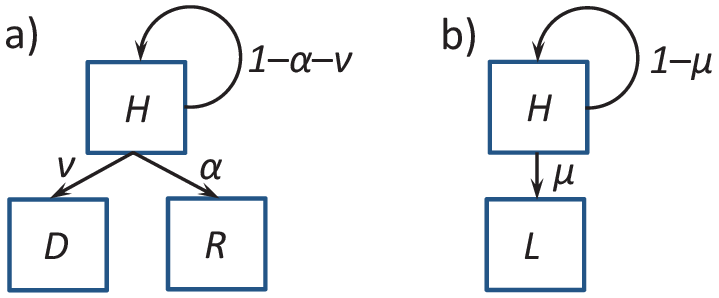}
\caption{a) The basic Markov model of mortality (`recovery/death lottery') with two
absorbing states {(states from which patients do not leave)}, `D' (death) and `R'
(recovery). b) The `lottery of transfer' (from the TARN network)  with one absorbing
state `L' (`left'). The transition probabilities $\alpha=\alpha(t,s)$, $\nu=\nu(t,s)$ and
$\mu(t,s)$ depend on the time after injury $t$ and on the state of the patient on the
first day after trauma presented by the values of attributes $s$.
\label{Fig:BasicMortalityModel}}
\end{centering}
\end{figure}

\begin{figure}
\begin{centering}
\includegraphics[width= 0.35\textwidth]{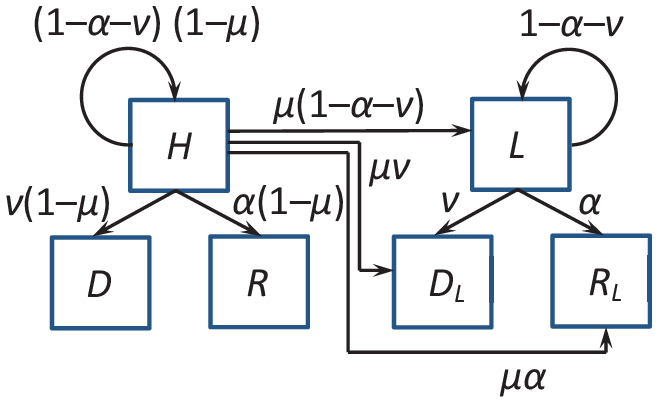}
\caption{The Markov model of mortality and transfer from TARN hospitals to hospitals
outside TARN for the limit case of `advanced transfer', when the lottery of transfer
(Fig.~\ref{Fig:BasicMortalityModel} b) occurs every day {\em before} the lottery of
survival (Fig.~\ref{Fig:BasicMortalityModel} a). It has six states: `H' (an alive patient
in a TARN hospital), `L' (an alive patient in a hospital outside TARN), `D' (death in a
TARN hospital), `${\rm D_L}$' (death in a hospital outside TARN), `R' (recovery of a
patient in  a TARN hospital) and `${\rm R_L}$' (recovery of a patient in a hospital
outside TARN). Four of them are absorbing: `D', `${\rm D_L}$', `R', and `${\rm R_L}$'.
The transitions from H to ${\rm D_L}$ and ${\rm R_L}$ are superpositions of the same day
transitions: ${\rm H \to L \to D_L}$ and $\rm H\to L \to
R_L$\label{Fig:MortalityModelBefore}}
\end{centering}
\end{figure}

\begin{figure}
\begin{centering}
\includegraphics[width= 0.35\textwidth]{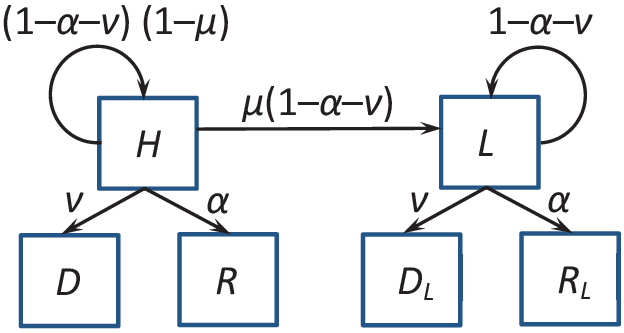}
\caption{The Markov model of mortality and transfer from TARN hospitals to hospitals
outside TARN for the limit case of `retarded transfer', when the lottery of transfer
(Fig.~\ref{Fig:BasicMortalityModel} b) occurs every day {\em after} the lottery of
survival (Fig.~\ref{Fig:BasicMortalityModel} a). \label{Fig:MainMarkovModelAfter} It has
the same states as the model with advanced transfer (Fig.~\ref{Fig:MortalityModelBefore})
but different transition probabilities.}
\end{centering}
\end{figure}

We propose a system of Markov models for evaluation of mortality in trauma datasets. In
these models each day each patient can participate  in two `lotteries'
(Fig.~\ref{Fig:BasicMortalityModel}). The first lottery (recovery/death),
Fig.~\ref{Fig:BasicMortalityModel} a, has three outcomes: `R' (recovery), `D' (death),
and `H' (remains in a TARN hospital). The second lottery (of transfer),
Fig.~\ref{Fig:BasicMortalityModel} b, has two outcomes: `H' (remains in a TARN hospital)
and `L' (transfer from the TARN hospital to a hospital or `other institution' outside
TARN). The probabilities of outcomes depend on the {\em time from the injury} $t$ and on
the {\em state of the patient} after injury $s$. It is important to stress that $s$ in our models characterizes the state of the patient
on the first day after trauma and may include severity, type of injury
(blunt/penetrating), localization of traumas, age, gender,  airway status, systolic and
diastolic blood pressure, etc, but cannot change in time.

The description of state $s$ may
vary in the level of {detail} depending {on} the available information. We have {fitted} and
tested two models based on the severity of trauma: the maximal severity model and the
(binned) NISS model. In Section~\ref{Sec:refine} we demostrate that 
it is necessary to refine the model and to include the age group in $s$ for low severities. 
For different purposes the mortality model can include more detail.

The lotteries (Fig.~\ref{Fig:BasicMortalityModel}) do not commute. We consider two limit
cases: `advanced transfer' (Fig.~\ref{Fig:MortalityModelBefore}) and `retarded transfer'
(Fig.~\ref{Fig:MainMarkovModelAfter}). In models with advanced transfer the lottery
of transfer Fig.~\ref{Fig:BasicMortalityModel} b) {each} day precedes the lottery of
recovery/death (Fig.~\ref{Fig:BasicMortalityModel} a). In models with retarded
transfer, conversely, the lottery of recovery/death {precedes the} lottery of transfer.

These two models are important because many  other much more general Markov models are
between them in the following exact sense. It is a very strong assumption that every day
there are two steps only: the recovery/death lottery and the transfer lottery. It may be
more realistic to assume that every day there are many `fractional steps' of
recovery/death and of transfer from TARN and the result of the day is the aggregate
result of all of these fractional steps. Assume {that the events} of recover, death and
transfer are sampled for every day after injury $t$ from {a number} $M$ consecutive
random choices with probabilities $\alpha_i,\nu_i$ for recovery/death and $\mu_i$ for
transfer out of TARN {($i=1,\ldots,M$), and} this chain of choices is Markovian (the
choices for a patient do not depend on the previous choices directly but only on the
current state, H, R or L). It is non-stationary because the transition probabilities depend 
on time. They are different for different days after injury.

This sequence of choices is displayed as a
sequence of fractional steps:
\begin{equation*}
\begin{split}
\mbox{recovery/death}_1&\to\mbox{transfer}_1\to\ldots \\ &\to
\mbox{recovery/death}_M\to\mbox{transfer}_M.
\end{split}
\end{equation*}

The probability of in-TARN death in the above model of sequential choice, {on} a
given day after trauma is
$$\nu_1+\nu_2 (1-\alpha_1-\nu_1)(1-\mu_1)+\ldots +\nu_M
\prod_{i=1}^{M-1}(1-\alpha_i-\nu_i)(1-\mu_i).$$

{Similarly, the probability for recovery is}
$$\alpha_1+\alpha_2 (1-\alpha_1-\nu_1)(1-\mu_1)+\ldots +\alpha_M
\prod_{i=1}^{M-1}(1-\alpha_i-\nu_i)(1-\mu_i).$$

{Finally, the probability of transfer to a hospital outside of TARN is}
\begin{equation*}
\begin{split}
\mu_1 (1-\alpha_1-\nu_1)&+\mu_2(1-\mu_1)(1-\alpha_1-\nu_1)(1-\alpha_2-\nu_2)\\
&+\ldots +\mu_M\prod_{i=1}^{M-1}(1-\mu_i) \prod_{j=1}^M(1-\alpha_j-\nu_j) {.}
\end{split}
\end{equation*}

The probabilities $\alpha_i,\nu_i$ for the fractional steps should be consistent with the
daily probabilities $\alpha, \nu$: if there is no transfer then the resulting
probabilities of {recovery or death} should be the same:

\begin{equation}\label{alphanuconditiom}
\begin{split}
&\alpha_1+\alpha_2 (1-\alpha_1-\nu_1)+\ldots +\alpha_M
\prod_{i=1}^{M-1}(1-\alpha_i-\nu_i)=\alpha , \\
&\nu_1+\nu_2 (1-\alpha_1-\nu_1)+\ldots +\nu_M \prod_{i=1}^{M-1}(1-\alpha_i-\nu_i)=\nu. \\
&\mbox{{Also,} } \prod_{i=1}^M(1-\alpha_i-\nu_i)=1-\alpha-\nu.
\end{split}
\end{equation}
Similarly, for $\mu_i$ we get the {conditions}
\begin{equation}
\begin{split}\label{mucondition}
\mu_1+\mu_2(1-\mu_1)+\ldots+\mu_M \prod_{i=1}^{M-1}(1-\mu_i)&=\mu \\
\mbox{ {and} } \prod_{i=1}^M(1-\mu_i)&=1-\mu {.}
\end{split}
\end{equation}

\begin{proposition}
The probability {of in-TARN} death in the described model of sequential choice for
every day after trauma is between {the} probabilities for the Markovian model with
advanced transfer (Fig.~\ref{Fig:MortalityModelBefore}) and the Markovian model with
retarded transfer (Fig.~\ref{Fig:MainMarkovModelAfter}):
\begin{equation}\label{BetweenIneq}
\begin{split}
\nu(1-\mu) \leq \nu_1&+\nu_2 (1-\alpha_1-\nu_1)(1-\mu_1)+\ldots \\&+\nu_M
\prod_{i=1}^{M-1}(1-\alpha_i-\nu_i)(1-\mu_i)\leq \nu .
\end{split}
\end{equation}
\end{proposition}
\begin{proof}
According to conditions (\ref{alphanuconditiom}), (\ref{mucondition}),
\begin{equation}\label{Firsttrans}
\begin{split}
&\nu(1-\mu)\\&=\left[\nu_1+\nu_2 (1-\alpha_1-\nu_1)+\ldots +\nu_M
\prod_{i=1}^{M-1}(1-\alpha_i-\nu_i)\right] \\ & \quad \times \prod_{i=1}^M(1-\mu_i).
\end{split}
\end{equation}
Notice that for every $j$ ($1\leq j \leq M$),
$$ \prod_{i=1}^M(1-\mu_i) \leq \prod_{i=1}^j (1-\mu_i) $$
because $0\leq 1-\mu_i \leq 1$ for all probabilities $\mu_i$.
Therefore,
$$\nu_j\prod_{i=1}^j (1-\alpha_i-\nu_i)\prod_{k=1}^M (1-\mu_k) \leq \nu_j\prod_{i=1}^j (1-\alpha_i-\nu_i)(1-\mu_i)$$
and the following inequality holds
\begin{equation}
\begin{split}
&\left[\nu_1+\nu_2 (1-\alpha_1-\nu_1)+\ldots +\nu_M
\prod_{i=1}^{M-1}(1-\alpha_i-\nu_i)\right] \\ & \times \prod_{i=1}^M(1-\mu_i) \\ & \leq
\nu_1+\nu_2 (1-\alpha_1-\nu_1)+\ldots +\nu_M \prod_{i=1}^{M-1}(1-\alpha_i-\nu_i).
\end{split}
\end{equation}
The left  inequality in (\ref{BetweenIneq}) is proven. The right inequality in
(\ref{BetweenIneq}) follows from condition (\ref{alphanuconditiom}) because for every
{product}
$$\nu_j\prod_{i=1}^j (1-\alpha_i-\nu_i)(1-\mu_i) \leq \nu_j\prod_{i=1}^j (1-\alpha_i-\nu_i) {.}$$

\end{proof}

The {proofs} of the following propositions {are} very similar

\begin{proposition}
The probability {of in-TARN} recovery in the described model of sequential choice for
every day after trauma is between {the} probabilities for the Markovian model with
advanced transfer (Fig.~\ref{Fig:MortalityModelBefore}) and the Markovian model with
retarded transfer (Fig.~\ref{Fig:MainMarkovModelAfter}):
\begin{equation}\label{BetweenIneqalpha}
\begin{split}
\alpha(1-\mu) \leq \alpha_1&+\alpha_2 (1-\alpha_1-\nu_1)(1-\mu_1)+\ldots \\&+\alpha_M
\prod_{i=1}^{M-1}(1-\alpha_i-\nu_i)(1-\mu_i)\leq \alpha  {.}
\end{split}
\end{equation}
\end{proposition}\begin{flushright}$\square$\end{flushright}

\begin{proposition}
The probability  {of transfer} outside TARN in the described model of sequential choice
for every day after trauma is between  {the} probabilities for the Markovian model with
advanced transfer (Fig.~\ref{Fig:MortalityModelBefore}) and the Markovian model with
retarded transfer (Fig.~\ref{Fig:MainMarkovModelAfter}):
\begin{equation}\label{BetweenIneqmu}
\mu (1-\alpha-\nu) \leq \sum_{j=1}^M \mu_j (1-\alpha_j-\nu_j)
\prod_{i=1}^{j-1}(1-\alpha_i-\nu_i)(1-\mu_i)\leq \mu {.}
\end{equation}
\end{proposition}\begin{flushright}$\square$\end{flushright}

\subsection{Transition probabilities and their evaluation}

In the  {above} models (Figs.~\ref{Fig:MortalityModelBefore}  {and}
\ref{Fig:MainMarkovModelAfter}),  death and recovery of the transferred patients have
the same probabilities  {as} for the patients of TARN hospitals. These probabilities are
defined by the state of the patient $s$  and  by the time after injury. Of course, in
reality there is often a hope that the transfer will improve the situation and the
probability of death will decrease for the same state of the patient. Nevertheless, in
this paper we will neglect the changes of probabilities after transfer (just because
we have no {\em sufficient reason} for such a change). Of course, these models could be
extended to include the changes of mortality for transferred patients, if necessary.

Another question is the definition of $s$. Which attributes should be included in the `state' for the models
(Figs.~\ref{Fig:MortalityModelBefore}, \ref{Fig:MainMarkovModelAfter})? To motivate this choice, we should take into account
two considerations:
\begin{enumerate}
\item The models will be used to  {analyse data} with unknown outcomes. Trauma cases
    with missed outcomes make up 10-12\% of the dataset. Therefore,
    {an} error  {of} 10\% {in}
    mortality for data with unknown outcomes will cause  {an} error
    {of} $\sim$1\% {in} mortality for the whole
    dataset and it is possible to use relatively coarse models  {(see
    below)}.
\item  The description of the state $s$ should include attributes whose values are known for
    {a significant} majority of cases. This is especially important because
    {for cases} with unknown outcomes many of {the}
    attributes are often also unknown ({a} more detailed analysis of
    data with missed attributes is presented in the next {section}).
\end{enumerate}

Formally, there are many possibilities for defining $s$. It {could} include the initial
state after trauma (characteristics of injury and coma status, for example), age, gender,
the current state ($t$ days after trauma), fragments of history, etc.  {The set of the auxiliary variables which may be selected  as potential sources of information could be much 
larger. For example, for creation of the model for imputing missing physiological data in the National Trauma Data Bank (NTDB), USA,  the following variables were used: gender, age,  components of Glasgow Coma Status, the maximum AIS or ICISS (and, separately,   the maximum AIS or ICISS for head injuries), injury type (penetrating, blunt), prehospital intubation, duration of mechanical ventilation, tests for alcohol and drugs, etc. \cite{Moore2009}. Nevertheless, even the simple models identified in our paper solve the problem of mortality correction quite well. The extention of the set of variables will not include essential methodological novelty and may be performed easily for sufficiently large datasets.} For our purposes,
we select, identify and compare three coarse models:
 {\begin{model}[The coarsest model] $s=\emptyset$. \label{mod1} \end{model}
\begin{model}[The maximal severity model] $s=$the maximal severity score (an integer from 1 to 6).\label{mod2} \end{model}
\begin{model}[The binned NISS model ]We use seven bins: NISS=1-3, 4-8, 9, 10-16, 17-24, 25-35, 36+;
    $s$  is the bin number (7 values). The bins for $s=2,\ldots,  7$ have approximately equal depth whereas the
    bin with $s=1$ (NISS=1-3) is much smaller. (For this first bin we found that the model should be supplemented by age.)\label{mod3} \end{model}}

We observe that the cases with maximal severity 1 (or NISS=1-3, which is the same)
are very special. First of all, the age distributions in this group for the `Available
W30D' and the `OUT30' subgroups are very different (Fig.~\ref{Fig:LowSeverityAges}). If
we do not take into account this difference then we overestimate mortality in this
group. The necessary refinement of the model with isolation of elderly patients with low
severity of trauma is presented in Section~\ref{Sec:refine}.
\begin{figure}
\begin{centering}
\includegraphics[width= 0.45\textwidth]{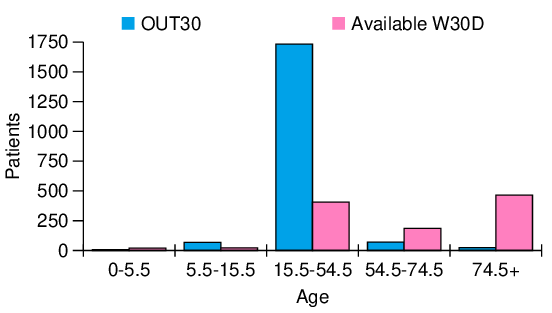}
\caption{ \label{Fig:LowSeverityAges} Age distributions for two groups of low severity
cases (NISS bin 1-3). The age distribution  for the low severity patients in TARN
(`Available W30D' AND NISS=1-3) for age binned in five bins (0-5.5; 5.5-15.5; 15.5-54.5;
54.5-74.5;  $>$74.5) has clear maximum for elderly patients (age $>$74.5), whereas the
absolute majority of the the low severity patients which left TARN without registered
outcome (`OUT30' AND NISS=1-3) belong to the group with age 15.5-54.5. }
\end{centering}
\end{figure}

Our approach may be combined with any stochastic model for early outcome prediction (see,
for example, \cite{Lavoie2004,Shoemaker2006,Brockamp2013}).

For the finite set of $s$ values, evaluation of all the coefficients $\alpha(t,s)$,
$\nu(t,s)$, and $\mu(t,s)$ is a particular case of a standard statistical problem of {\em
proportion estimate} for each given value of $s$; we use the Wilson score interval (CI)
\cite{Wilson1927}:
\begin{equation}\label{WCI}
\frac{1}{1+\frac{z^2}{n}}\left[\hat{p}+\frac{z^2}{2n}\pm z\sqrt{\frac{\hat{p}(1-\hat{p})}{n}+\frac{z^2}{4n^2}}\right],
\end{equation}
where $\hat{p}$ is the coefficient estimate, $z$ is the error percentile ($z=1.96$ for
{the} 95\% confidence interval), and $n$ is the number of degrees of freedom (for a
dataset without weights this is just the sample size).

For the coarsest model (Model~\ref{mod1}) the fraction of {patients transferred outside TARN is} {11.65\%.
This} is just the fraction of {patients transferred (within 30 days after injury) in}
Table~\ref{Table:2}. The 95\% CI (\ref{WCI}) for this fraction is 11.5--11.8\%. For the
maximal severity (Table~\ref{Table:3}) (Model~\ref{mod2}) and the binned NISS (Table~\ref{Table:4}) ((Model~\ref{mod3}) models
the fraction of {patients transferred outside TARN depends} on $s$ (bins) and the CI in
each {bin is} larger than for the total fraction in the coarsest models (Model~\ref{mod1}). Nevertheless,
{the} CIs for different bins in {these} models do not intersect  (the only exclusion is
the CI for the smallest bin, maximal severity 6, in the maximal severity model (Model~\ref{mod2}),
Table~\ref{Table:3}). In particular, this means that the probability {of} transfer
outside TARN hospitals {depends strongly} on the trauma severity.

\begin{table*}{\footnotesize\begin{center}
\caption{\label{Table:3} Sizes of bins and fractions of transfer out of TARN (within 30
days after injury) for the maximal severity models.}
\begin{tabular}{ccccc}
\hline Max severity & OUT30 & Total & Fraction of OUT30 & 95\% CI \\ \hline 1 & 1,905 &
3,005 & 63.39\% & 61.66--65.10\% \\ \hline 2 & 3,094 & 35,109 & 8.81\% & 8.52--9.11\% \\
\hline 3 & 6,203 & 77,518 & 8.00\% & 7.81--8.20\% \\ \hline 4 & 4,535 & 29,603 & 15.32\%
& 14.91--15.73\% \\ \hline 5 & 3,542 & 20,175 & 17.56\% & 17.04--18.09\% \\ \hline 6 & 10
& 149 & 6.71\% & 3.88--11.72\% \\ \hline
\end{tabular}
\end{center}}
\end{table*}

\begin{table*}{\footnotesize\begin{center}
\caption{\label{Table:4} Sizes of bins and fractions of patients  transferred to a
hospital or institution (or unknown destination) (within 30 days after injury) for the
binned NISS Models~\ref{mod3} .}
\begin{tabular}{ccccc}
\hline NISS bin & OUT30 & Total & Fraction of OUT30 & 95\% CI \\ \hline 1-3 & 1,905 &
3,005 & 63.39\% & 61.66--65.10\% \\ \hline 4-8 & 2,078 & 24,982 & 8.32\% & 7.98--8.67\%
\\ \hline 9 & 2,159 & 36,722 & 5.88\% & 5.64--6.12\% \\ \hline 10-16 & 2,710 & 29,237 &
9.27\% & 8.94--9.61\% \\ \hline 17-24 & 2,882 & 25,074 & 11.49\% &11.11--11.89\% \\
\hline 25-35 & 3,603 & 23,557 & 15.29\% & 14.84--15.76\% \\ \hline 36+ & 3,952 & 22,982 &
17.20\% & 16.71--17.69\% \\ \hline
\end{tabular}
\end{center}}
\end{table*}

For each value of $s$ and time after injury $t$ the following quantities are found for
the {analysed} dataset:
\begin{itemize}
\item $H(t,s)$ -- the number of patients in state $s$ registered as alive in a TARN
    hospital at any time during day $t$ after injury (in this number we include the
     {patients which} have stayed at a TARN hospital  {during day} $t$ after injury,
    the patients who have died  {on} this day in a TARN hospital, have been  {discharged, or} have been transferred outside TARN  {on} this day);
\item $\Delta D(t,s)$ -- the number of patients in state $s$ {who} died in TARN hospitals
    {on
    day} $t$ after injury;
\item $\Delta R(t,s)$ -- the number of patients in state $s$  {who} recovered (discharged
    to one
    of the final recovery destinations) in TARN hospitals {on day} $t$ after
    injury;
\item $\Delta L(t,s)$ -- the number of patients  in state $s$ {who} transferred out of TARN
    hospitals to other hospitals, institutions or unknown destinations  {on day} $t$ after
    injury.
\end{itemize}
Just for control, the following identity should hold: $H(t+1,s)=H(t,s)-\Delta
D(t,s)-\Delta R(t,s)-\Delta L(t,s)$ because state $s$ in our models does not change in
time.

For the model with advanced transfer from TARN hospitals the coefficients are defined
following the scheme presented in Fig.~\ref{Fig:MortalityModelBefore}:
\begin{equation}\label{coeffbefore}
\begin{split}
&\mu(t,s)=\frac{\Delta L(t,s)}{H(t,s)}; \; \nu(t,s)=\frac{\Delta
D(t,s)}{(1-\mu(t,s))H(t,s)}; \; \\ &\alpha(t,s)=\frac{\Delta R(t,s)}{(1-\mu(t,s))H(t,s)}.
\end{split}
\end{equation}

For the model with retarded transfer from TARN hospitals the coefficients are defined
following the scheme presented in Fig.~\ref{Fig:MainMarkovModelAfter}:
\begin{equation}\label{coeffafter}
\begin{split}
&\nu(t,s)=\frac{\Delta D(t,s)}{H(t,s)}; \; \alpha(t,s)=\frac{\Delta R(t,s)}{H(t,s)}; \;
\\ &\mu(t,s)=\frac{\Delta L(t,s)}{(1-\alpha(t,s)-\nu(t,s))H(t,s)}.
\end{split}
\end{equation}

\subsection{Evaluation of FOD}

Each model provides us {with} the {\em corrected FOD}. We use the basic assumption that
{the} probability {of dying} at  time $t$ after injury depends on $s$ but is the same
inside and outside TARN. For each $t$ and $s$ we define the {\em specific cumulative FOD}
(scFOD$(t,s)$) as the fraction of patients with state $s$ {who} died during the time
interval $[1,t]$:
\begin{equation}\label{scFOD}
\begin{split}
\mbox{scFOD}(t,s)=\nu(1,s)&+\nu(2,s)(1-\alpha(1,s)-\nu(1,s))+\ldots
\\&+\nu(t,s)\prod_{i=1}^{t-1}(1-\alpha(i,s)-\nu(i,s)).
\end{split}
\end{equation}

The cumulative FOD at time $t$ (cFOD$(t)$) for the whole model (for all $s$ together) is
\begin{equation}\label{cFOD}
\mbox{cFOD}(t)=\frac{\sum_s \mbox{scFOD}(t,s) H(1,s)}{H_0} {,}
\end{equation}
where $H_0=\sum_s H(1,s)$ is the total number of patients in our dataset (in our case
study, $H_0=165,559$).

The functions cFOD$(t)$ and {scFOD$(t,s)$ for all $s$}, grow monotonically with $t$.

If we define the final outcome as survival or death {within} 30 days after injury then
the target value is FOD= cFOD(30).

Let us compare two following na\"{\i}ve approaches to {the} handling of missing outcomes with the
{Markov models we have created.}
\begin{itemize}
\item {\em Available case analysis.} Just delete all of the 19,289 cases with the
    outcome     `Transferred OUT OF TARN within 30 days after injury' from the dataset. In the remaining cases all
    outcomes are known and the FOD is the ratio $\frac{\mbox{Dead (W30D)}}{\mbox{Total}}$ in the reduced dataset.
\item {\em Consider all transferred patients as alive.} In this case, the total number of
    patients does not change and the FOD is the  ratio $\frac{\mbox{Dead (W30D)}}{\mbox{Total}}$,
     where the number `Dead (W30D)' is the same but the number `Total' is calculated for the whole original
    dataset (Table~\ref{Table:2}).
\end{itemize}

\begin{remark} If {we} apply available case analysis then {none of the numbers $\Delta D(t,s)$ and $\Delta R(t,s)$
change} but the numbers $H(t,s)$ of the patients {in TARN} will decrease for all $t$ and $s$
(or do not change if there is nothing to delete).  {The corresponding mortality
coefficients {$\nu(t,s)$ } will be larger than the coefficients (\ref{coeffbefore}),}
(\ref{coeffafter}) for all the Markov models considered before. This means that the MCAR
(Missing Completely At Random) approach to missed outcomes  always overestimates
mortality, while the second na\"ive approach (`Consider all transferred patients as
alive') always underestimates mortality.
\end{remark}

We {have} created six Markov models for mortality of transferred patients. They differ by
the state variable $s$ (the coarsest model without $s$, Model~\ref{mod1}, the maximal severity model with six states, Model~\ref{mod2} and the binned NISS model with seven states, Model~\ref{mod3}) and by the order of the
`recovery/death' and `transfer' lotteries (Fig.~\ref{Fig:BasicMortalityModel}). In
Table~\ref{Table:5} we compare the {mortality evaluated} by these {models, and} by the two
na\"ive models. We can see that the difference between all {of} our Markov models is not
{significant; we} cannot reject the hypothesis that they coincide with {any} one of them
($p$-value is between 0.20 and  { 0.56}). Both {of} the na\"ive models differ significantly
from all {of} the six Markov models. The difference between the na\"ive models is also
significant. All the values of mortality predicted by the Markov models belong to the interval
$(6.77\%,6.91\%)$. The average of {the six} Markovian predictions is 6.84\%. None of the
Markov model predictions differ significantly from this average. Both {of} the na\"ive
predictions are significantly different.

\begin{table*}{\footnotesize\begin{center}
\caption{\label{Table:5}FOD for different models.  { Here, the $p$-value is the probability of  observing `by chance' equal or greater deviation of FOD from the value FOD=6.85\% given by the coarsest advanced model', under the condition that the expectations of  FOD is 6.85\%.} }
\begin{tabular}{ccccc}
\hline Model & Alive & Dead & FOD & $p$-value \\ \hline \hline
 Available case study & 135,733 & 10,537 & 7.20\% & $1.3\times10^{-8}$ \\
 \hline All transferred are alive & 155,022 & 10,537 & 6.36\% & $5.0\times10^{-15}$ \\ \hline
\hline Coarsest advanced & 154,217 & 11,342 & 6.85\% & 1.00 \\
\hline Coarsest retarded & 154,350 & 11,209 & 6.77\% & 0.20\\
 \hline Max severity, advanced& 154,120 & 11,439 & 6.91\% & 0.34\\
 \hline Max severity, retarded & 154,266 & 11,293 & 6.82\% & 0.41 \\ \hline
NISS binned, advanced & 154,145 & 11,414 &  6.89\% & 0.48 \\ \hline
 NISS binned, retarded & 154,292 & 11,267 & 6.81\% & 0.57 \\ \hline
\end{tabular}
\end{center}}
\end{table*}

\subsection{Validation of the models on the excluded trauma cases: patients  transferred to TARN (`IN30')}

For each type of model the coefficients $\mu$, $\alpha$ and $\nu$ are evaluated using
the  dataset of 165,559 patients entering TARN in the first day of injury
(Fig~\ref{MainGroups}, Main Group). Let us test the models with evaluated coefficients we have described here 
on data we have not used before.
These data consist of the 16,693 cases who came to TARN hospitals more than one day after injury, 
which we deleted from the original set before modelling.  
This is a special and biased sample, `IN30' (see
Fig.~\ref{MainGroups}). We now apply the
models developed and identified in the previous {subsections} to {analyse this} sample.
We expect that there should be some similarity between the groups of patients transferred
{\em from TARN} (`OUT30') and the patients transferred {\em to TARN} (`IN30')
(Fig.~\ref{MainGroups}).  We do not expect quantitative coincidence of the results for
the groups `OUT30' and `IN30' because there is no precise symmetry between the patients
moved to TARN and the patients moved from TARN. The hospitals in TARN are those with a
special interest in trauma - in particular the large major trauma centers, so the
transfers in (mainly for acute specialist care) will not be the same as those transferred
out (mainly for complex rehabilitation, or special geriatric care, etc.).

Therefore, the estimated behavior of the mortality {of} the group transferred from TARN can be {\it qualitatively} validated
using the observed mortality in the group {who} moved to TARN.

We {consider survival during the first 30 days. Hence} we have to use the records which
correspond to this period only. There are 15,437 such records among the 16,693 in 
`IN30'.

In  {these} estimates of the FOD we explicitly use the empirical fluxes into and from
TARN hospitals. For each $t,s$ we have the following quantities:
\begin{itemize}
\item $L_{\rm in}(t,s)$ -- the number of patients in state $s$ which came to TARN {on
    day  $t$ after} injury;
\item $L_{\rm out} (t,s)$ -- the number of patients in state $s$ from `IN30' which were
    transferred from TARN on day $t$ after injury.
\item $h_{\rm IN30}(t,s)$ -- the number of patients in IN30 in state $s$ on day $t$
    after injury.
\item $D_{\rm IN30}(t,s)$ -- the number of deaths in TARN {of} the patients from IN30 in
    state
    $s$ {by}  day  $t$  after injury (cumulative).
\item $R_{\rm IN30}(t,s)$ -- the number of patients in `IN30' in state $s$ {who recovered
    by    day} $t$ after injury (cumulative).
\end{itemize}
We use the values $L_{\rm in}(t,s)$ and $L_{\rm out} (t,s)$ from the database, evaluate
$h_{\rm IN30}(t,s)$, $D_{\rm IN30}(t,s)$, and $R_{\rm IN30}(t,s)$ {for} every model and
then compare the resulting outcomes  (evaluated numbers of death in TARN of the patients
from `IN30' within 30 days of injury, $\sum_s D_{\rm IN30}(30,s)$) to empirical data
from TARN records.

For each model with advanced transfer the variables $h_{\rm IN30}(t,s)$, $D_{\rm
IN30}(t,s)$, and $R_{\rm IN30}(t,s)$ are evaluated by {recurrence} formulas:
\begin{equation}
\begin{split}
h_{\rm IN30}(t+1,s)&=[h_{\rm IN30}(t,s)+L_{\rm in}(t+1,s) \\-&L_{\rm out}
(t+1,s)][1-\alpha(t+1,s)-\nu(t+1,s)]; \\
R_{\rm IN30}(t+1,s)&=R_{\rm IN30}(t,s)+\alpha(t+1,s)\\ \times [&h_{\rm IN30}(t,s)+L_{\rm
in}(t+1,s)-L_{\rm out} (t+1,s)]; \\
D_{\rm IN30}(t+1,s)&=D_{\rm IN30}(t,s)+\nu(t+1,s) \\ \times [&h_{\rm IN30}(t,s)+L_{\rm
in}(t+1,s)-L_{\rm out} (t+1,s)] {,}
\end{split}
\end{equation}
with initial condition $$h_{\rm IN30}(0,s)=R_{\rm IN30}(0,s)=D_{\rm IN30}(0,s)=0.$$

For each model with retarded transfer the variables $h_{\rm IN30}(t,s)$, $D_{\rm
IN30}(t,s)$, and $R_{\rm IN30}(t,s)$ are evaluated by {recurrence} formulas:
\begin{equation}
\begin{split}
h_{\rm IN30}(t+1,s)=&h_{\rm IN30}(t,s)[1-\alpha(t+1,s)-\nu(t+1,s)]\\ &+L_{\rm
in}(t+1,s)-L_{\rm out} (t+1,s); \\
R_{\rm IN30}(t+1,s)=&R_{\rm IN30}(t,s)+\alpha(t+1,s)h_{\rm IN30}(t,s); \\
D_{\rm IN30}(t+1,s)=&D_{\rm IN30}(t,s)+\nu(t+1,s)h_{\rm IN30}(t,s) {,}
\end{split}
\end{equation}
with initial condition $$h_{\rm IN30}(0,s)=R_{\rm IN30}(0,s)=D_{\rm IN30}(0,s)=0.$$

For each model, the coefficients $\alpha(t,s)$ and $\nu(t,s)$ are evaluated using the
previously {analysed} dataset (without IN30) by formulas (\ref{coeffbefore}) {and}
(\ref{coeffafter}). The results are presented in Table~\ref{Table:6}.

\begin{table*}{\footnotesize\begin{center}
\caption{\label{Table:6}Comparison of the models with the empirical data about patients
from `IN30'.}
\begin{tabular}{cccccc}
\hline Model & Alive & Dead & Total & FOD & CI 95 \\ \hline
 Empirical data & 13,038.00 & 417.00 & 13,455.00 & 3.10\% & 2.82--3.41\% \\ \hline
 \hline Coarsest advanced & 12,834.55 & 620.45 & 13,455.00 & 4.61\% & 4.27--4.98\% \\
\hline Coarsest retarded & 12,933.67 & 521.33 & 13,455.00 & 3.87\% & 3.56--4.21\%\\
 \hline Max severity, advanced & 12,824.90 & 630.10 & 13,455.00 & 4.68\% & 4.34--5.05\%\\
 \hline Max severity, retarded & 12,920.71 & 534.29 & 13,455.00 & 3.97\% & 3.65--4.31\% \\ \hline
NISS binned, advanced &  12,885.93 & 569.07 & 13,455.00 & 4.23\% & 3.90--4.58\% \\ \hline
 NISS binned, retarded & 12,971.22 & 483.78 & 13,455.00 & 3.60\% & 3.29--3.92\% \\ \hline
\end{tabular}
\end{center}}
\end{table*}

We can see that all the models overestimate mortality in `IN30'. The available case
analysis demonstrates the worst performance (the relative error exceeds 100\% of
empirical mortality). Models with retarded transfer {perform better in this test}
than the models with advanced transfer. The NISS binned model with retarded transfer is
the best (the relative error in prediction of FOD is 16\% of  the empirical data and, at
least, the 95\% confidence intervals for the result of this model and for the empirical
data intersect). There exist further possibilities for improving the models presented but already the
relative error of $16\%$ for `IN30' in the estimation for the total database will give the
input in the relative error in the FOD $\lesssim$1\% (or absolute error
$\lesssim$0.07\%). That is much better than the errors of the available case evaluations
or of the approach `all are alive' to the evaluation of mortality of transferred
patients.

 {\subsection{Validation of the model for the mortality prediction in the `Available W30D' group of TARN patients on  `real death -- simulated transfer' data}}

 {The successful test on the group  `IN30' of patients transferred to TARN supports the approach developed in this work. Nevertheless, transfer {\it to} TARN hospitals differs from transfer {\it from} TARN qualitatively because of a qualitative difference between hospitals included and not included in TARN. In this section, we provide  additional validation of the Markov models on the mortality prediction in the `Available W30D' group of TARN patients with known outcomes (Fig.~\ref{FirstDayTransfer}). We created a statistical model for imitation of patient transfer and use the known outcomes. This means, we use `real death -- simulated transfer' data.}

 {\begin{itemize}
\item Firstly, using the main group, we evaluated the transfer probability for each day in hospital as a function of NISS for 7 NISS bins,  separately for age$<65$ and age$\geq 65$. For example, a histogram of the number of transferred patients for the first day after trauma is presented in Fig.~\ref{FirstDayTransfer}.
\item Secondly, we take the `Available W30D' group and separate it into the `training set' and `test set'.  Random selection of the patients for the test set models transfer from TARN using probabilities evaluated at the previous step utilising the real data.
\item Thirdly, we create a Markov chain model using the training set and test the mortality in the   whole `Available W30D' group, which was not given during the modelling.
\end{itemize}}

\begin{figure*}
\begin{centering}
\includegraphics[width= 0.9\textwidth]{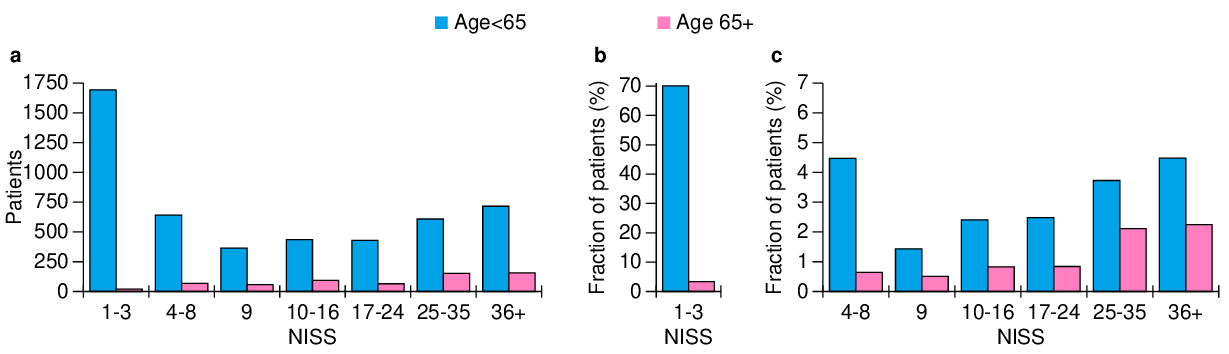}
 \caption{ \label{FirstDayTransfer}  Patients transferred from TARN during the first day after trauma (a) and the fraction of these patients (b,c) for seven NISS bins and two age groups (age$<65$ and age$\geq 65$). Note the scale difference between the fraction histogram for NISS=1-3 (b)  and for the other six bins (c).}
\end{centering}
\end{figure*}

\begin{table*}{\footnotesize\begin{center}
\caption{\label{Table:SimTransferTest}Test of the Markov `NISS binned, retarded' model on the `real death -- simulated transfer' data from `Available W30D' group}
\begin{tabular}{ccccc}
\hline 
\multirow{2}{*}
  & Fraction of & FOD (All & FOD (Available & FOD (Markov \\
  & transferred & alive)    & case study)      & model) \\ \hline
 Min & 9.84 & 6.83 & 7.59 & 7.18  \\ 
 \hline Max & 10.12 &  6.90 & 7.67 & 7.24\\
\hline Mean & 10.00 &6.86 & 7.63 & 7.21\\
 \hline St. deviation & 0.00628& 0.0140 & 0.0153 & 0.0147\\
 \hline Width  of 95\% CI & 0.0123 & 0.0274 & 0.0299 & 0.0285 \\ \hline
\end{tabular}
\end{center}}
\end{table*}

 {The random separation of the `Available W30D' group into training and test sets was performed 100 times. We evaluated the mortality for each such separation by two na\"ive models (available case study and `all transferred alive’ assumption) and the Markov `NISS binned, retarded’. The results were compared in Table~\ref{Table:SimTransferTest}. The fraction of death in the whole `Available W30D' group is 7.2\%. `Available case study' overestimates mortality (all the mortality values given by this approach in 100 trials are in the interval [7.59\%,7.67\%], which does not even  include the true value 7.2\%), the `all transferred are alive' hypothesis underestimates mortality (all the mortality values given by this approach in 100 trials are in the interval [6.83\%,6.86\%], which also does not include the true value 7.2\%). All the values given by the Markov `NISS binned, retarded’ model belong to the interval [7.18\%,7.24\%] around the true value with mean 7.21\% and standard deviation 0.0146\%. The relative error of this mortality prediction is small. It is less than 0.003 (or 0.3\%). This test on the `real death -- simulated transfer' data demonstrates the performance of the proposed method. }

\section{Model refinement \label{Sec:refine}}

We use a coarse model based on the severity of trauma for the evaluation of FOD in the group `OUT30'.
The reason for selection of such a coarse model is that a fraction of cases in this `OUT30' cohort is relatively small 
with respect to the `Available W30D' cases. As we 
can see from Table~\ref{Table:4}, this fraction is relatively small in all cells except
small severities with NISS=1-3 (see also Fig.~\ref{FirstDayTransfer} for the first day transfer). For refinement of the Markov model for this cell, we
compare the age structure of the `Available W30D' and the `OUT30' fractions of this
severity bin (Fig.~\ref{Fig:LowSeverityAges}). We see that the fraction of elderly patients
with low severities in TARN hospitals is high, whereas for patients transferred from
TARN this fraction is much lower. Mortality in the group of patients 74.5+ is much higher
than in the adult group, therefore the model overestimates mortality in the low severity states.
To refine the model let us use two cells for low severity: `NISS 1-3 y' (NISS bin 1-3 and
age $<54.5$) and `NISS  1-3 o' (NISS bin 1-3 and age $>54.5$). This refined model gives a 
significantly different FOD for NISS 1-3. In the cell `NISS 1-3 y' the corrected FOD is
0.54\% and in the cell `NISS  1-3 o' it is 4.08\% (almost eight times greater). The
corrected  overall FOD for NISS 1-3 is 1.42\% versus 2.68\% in the NISS retarded model
without the above refinement.

The effect of the refinement on the FOD for trauma cases is less because the fraction of
traumas with NISS severity 1-3 is relatively small (2.0\%).  For the refined model with
retarded transfer the FOD for transferred patients decreases from 3.79\% (retarded
transfer NISS model) to 3.59\% and the total fraction of death is changed from 6.81\% to
6.78\% (compare to Table~\ref{Table:5}).

\section{Weighting adjustment of {death} cases for further analysis \label{weighting}}

Single imputation of missed values does not reflect the uncertainty in data properly.
From the probabilistic point of view, a datapoint with missed values should be considered
as a conditional probability distribution of the form $$\mathbf{P}(\mbox{missed values }|\mbox{
known values}).$$ Two approaches utilize this idea the {\em multiple imputation} and 
{\em weighting adjustment}.

In the { multiple imputation} approach  several replicas of the database are created,
which differ in the imputed values \cite{Rubin1987,Rubin1996,Graham2007,Sterne2009}. The
distribution of this values should reflect the conditional means and conditional
variances of the imputing attributes. It is not completely clear, how many imputations
should be generated. Rubin claims that ``typically as few as five multiple imputations
(or even three in some cases) is adequate under each model for nonresponse''
\cite{Rubin1996}. Nevertheless, more recently, Graham et al produced practical
recommendations for selection of number of imputations $m$ and demonstrated that a
reasonable choice is $m\geq 20$ and for some cases $m=100$ is not enough
\cite{Graham2007}.  The multiple imputation algorithms are
implemented in the standard statistical software \cite{Royston2004}. 
Sterne et al \cite{Sterne2009} discussed use and misuse of imputation
in epidemiological and clinical research and tried to produce a standard for reporting of
handling of missed data in medical research.

 {It should be stressed that the risk prediction models which used data with gaps and rely on multiple imputation can be misleading, especially with many predictor variables \cite{Trickey2013}. Recently,  it was demonstrated that sensitivity analysis may be more informative than multiple imputation for study of the influence of missing data on risk prediction \cite{Trickey2013}.}

The weighting adjustment approach substitutes a datapoint with missed values by a set of
additional weights on the complete datapoints \cite{Little1988,Kalton1998,Little2003}.
The simplest version of this approach is the {\em cell weighting adjustment}. This 
follows the assumption that complete datapoints within a cell represent the incomplete
datapoints within that cell. An incomplete datapoint within the cell is substituted by
the equidistribution on the complete datapoints there. Of course, cell weighting can
inflate the variances for large cells. In this section, we use cell weighting
adjustments for the handling of missed outcomes. Cells are defined by state $s$ and the
outcome.

We will use the database for evaluation of the death risk for trauma patients. The `Main Group'
selected for further analysis includes the `OUT30' subgroup with 19,289 data cases transferred from TARN
hospitals within 30 days after injury (Fig~\ref{MainGroups}). The targeted outcome (alive
or dead within 30 days after injury) is unknown for these patients. Data without outcome
cannot be used {for risk} evaluation and should be deleted. Let us call the result of
deletion the {\em truncated} database. It is demonstrated in the previous sections {that}
the simple removal of the cases with unknown outcome  shifts the risk estimates;  the
proportion of Dead and Alive outcomes in the truncated database differs from reality and
the risk is overestimated (the pessimistic evaluation). This bias may be compensated by
reweighting of the cases with known outcomes. There are 146,270 such `Available W30D'
cases. In this subsection we estimate weights $w(t,s)$ that should be assigned to the
cases of death on day $t$ after injury with state $s$ to hold the probability of death
for {the} truncated database. For the estimation of the proper FOD that should be kept we
use the Markov model of mortality based on binned NISS (Model~\ref{mod3}) with delayed transfer out of TARN
(after selection dead and recovered patient, see Fig.~\ref{Fig:MainMarkovModelAfter}).
{This model} demonstrates the best verification results (Table~\ref{Table:6}) and is the most 
plausible from the common sense point of view.

{According} to the model, the probability of the patient in state $s$ dying on day $t$
after injury is evaluated as
$$p_d(t,s)=\frac{\Delta D (t,s)+\Delta D_L (t,s)}{H_0(s)} {,}$$
where $H_0(s)=H(1,s)$ is the initial number of patients in state $s$ {on} the first day
after injury. For the truncated data with weights this probability is evaluated as
the ratio of the sums with weights:
\begin{equation}
p_d^{w}(t,s)=\frac{w(t,s)\Delta D (t,s)}{H_0^w(s)},
\end{equation}
 where
\begin{equation}
{H_0^w(s)}=H(31,s) + R(30,s)+\sum_{t=1}^{30} w(t,s)\Delta D(t,s)
\end{equation}
and the superscript $w$ corresponds to the truncated dataset with weights. The numbers
$H(t,s)$, $R(t,s)$ and  $\Delta D(t,s)$ are the same for the original and truncated
datasets.

The probability {of dying} within 30 days from injury is evaluated as {the} proportion
of deaths (we use the model to find $D_L(30,s)$)
$$p_d(s)=\frac{D(30,s)+D_L(30,s)}{H_0(s)} {.}$$ {For} the truncated database $p_d(s)$ is evaluated as {the}
proportion {of weighted} deaths: $$p^w_d(s)=\frac{\sum_{t=1}^{30} w(t,s)\Delta
D(t,s)}{H_0^w(t,s)}. $$ This should be the same number. Therefore, the weighted sum of
deaths for the truncated database is:
$$\sum_{t=1}^{30} w(t,s)\Delta D(t,s)=\frac{p_d(s)}{1-p_d(s)}(H(31,s)+R(30,s)) {.}$$
{The} last expression in the {brackets} is just the number of `Alive within 30 days' outcomes.
Immediately we get
$$H_0^w(s)= {\frac{1}{1-p_d(s)}(H(31,s)+R(30,s))}.$$
{The formula} for the calculation of the weights of death cases in the truncated database is
\begin{equation}
w(t,s)= \frac{p_d(t,s)H_0^w(s) }{\Delta D(t,s)}.
\end{equation}

The weighting procedure changes the number of effective degrees of freedom can affect
the statistical power of the dataset but for the TARN dataset this change is rather minor. 
For example, for the standard problem of the
evaluation of the confidence interval in the proportion estimate the number  of degrees
of freedom $n_w$ in the weighted database with weights $w_i$ is
\begin{equation}\label{WeightedDOF}
n_w=\frac{\left(\sum_i w_i\right)^2}{\sum_i{w_i^2}}.
\end{equation}
For our dataset $n_w=143,574.85$ and the number of Available W30D records is 146,270 (Fig.~\ref{MainGroups}). The
difference of degrees of freedom for the non-weighted and weighted datasets is less than
2\%.

\section{FOD and patterns of mortality}

The models we have developed  { allow} us to evaluate the FOD for various groups of patients. The rich TARN data 
give us the chance of  studying various special groups and detailed stratifications of the trauma cases: by the severities of various injuries in combined traumas, by the age of patients, and by time (day) after trauma. Each example below is supplemented by a medical commentary.

\subsection{Example: FOD as function of age}

The age distribution of
trauma cases and the dependence of FOD on age are shown in Fig.~\ref{Fig:Mortality}.
Here we find surprisingly high accuracy of the piecewise linear approximation of FOD for
adult and elderly patients with a jump in the slope at $\mbox{age}\approx 62$.

\begin{figure}
\begin{centering}
\includegraphics[width= 0.45\textwidth]{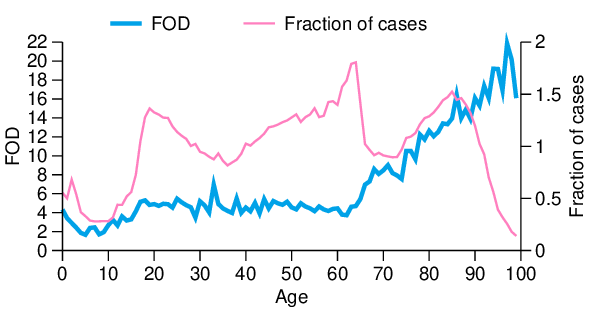}
 \caption{ \label{Fig:Mortality}Age distribution of trauma cases in `Available W30D' group
 and the FOD (corrected) as a function of age. The piecewise linear segmentation of
 FOD(age) has an obvious break point at $\mbox{age}\approx 62$.}
\end{centering}
\end{figure}

The number of cases per year in the dataset drops down at age 65 because for age$\geq 65$ some traumas 
are excluded from the database (see Fig.~\ref{MainGroups}). 

\paragraph{Medical commentary}

The increase in mortality with age is well established. Previous versions of the standard trauma outcome 
prediction system had two different models with an age cutoff at 55 years. More recent models have age
 as a weighted continuous variable with an interaction term between gender and age. There has been 
 a dramatic change in the trauma population over the last 10 years, with a rapid increase in the 
 number of older patients with major injury. Understanding the effects of age on trauma care and 
 adapting to a changing population will be a key challenge for trauma systems in the developed 
 world over the next 10 years.

\subsection{Example: FOD of combined traumas of various severity}

Evaluation of the severity of combined traumas is a classical problem. The very popular
solution is NISS -- sum of squares of three maximal severities, $s_1^2+s_2^2+s_3^2$
($s_1\geq s_2 \geq s_3$) (see, for example, \cite{Lavoie2004,Tay2004,Sullivan2003}). The
best severity score should give the best evaluation of mortality. This is a basic and
rather old idea for defining and comparing  trauma indices \cite{Sacco1977}. Of
course, it is possible to use three (or more) severities together as a multi-dimensional
trauma severity index (`severity profile' \cite{Sacco1988}) but the combination in one
index may be beneficial from different points of view.

The simplest method of combination is:
\begin{itemize}
\item Calculate FOD for every combination of severities for combined traumas for a large
    database;
\item Either use this  FOD instead of the severity score
\item Or find and use a convenient analytic approximation for this FOD (smoothed FOD).
\end{itemize}
Of course, such evaluation of probabilities for several input attributes were used by
many authors and compared to other approaches \cite{TARN8_2006,Brockamp2013}. In this
paper, we use TARN database and evaluate FOD of combined traumas as a function of three input attributes, three biggest
severity scores $s_1\geq s_2 \geq s_3$ (like in NISS). 

We use the  dataset of  146,270 `Available W30D' patients approached TARN during the first
day of injury and remained in TARN or were discharged to a final destination within the first 30
days after injury (Fig.~\ref{MainGroups}). 

Using our models, we calculate estimates with weights which take into account modeled mortality/survival of the
patients transferred from TARN and other patients with unknown outcomes. 
Results for the maximal severity $s_1=5$ are
presented in Table~\ref{Table:8}.  The available case analysis gives qualitatively the same results, hence,  the effects we observe  are not
generated by the reweighting procedure. 

\begin{table}{\footnotesize\begin{center}
\caption{\label{Table:8}FOD for the maximal severity $s_1=5$ and various $s_2$ and $s_3$ for
data after reweighting.}
\begin{tabular}{c|cccccc|}
\cline{2-7} & \multicolumn{6}{c|}{ $s_3$}  \\ \hline
 $s_2$ & 0 & 1 & 2 & 3 & 4 & 5   \\ \hline
 0 & 0.3590 &&&&& \\ \hline
 1 & 0.2324 & 0.2906 &&&& \\ \hline
 2 & 0.1566 & 0.1496 & 0.0791 &&&\\  \hline
 3 & 0.2466 & 0.2064 & 0.1315 & 0.1439 && \\ \hline
 4 & 0.2579 & 0.2881 & 0.1643 & 0.2105 & 0.3113 & \\  \hline
 5 & 0.4073 & 0.5668 & 0.4067 & 0.3666 & 0.4140 & 0.5908\\  \hline
\end{tabular}
\end{center}}
\end{table}

The results presented in Table~\ref{Table:8} seem to be counterintuitive:
FOD for combined injuries with severities $s_1=5$ and $1\leq s_2 \leq 4$ are less than
FOD for $s_2=s_3=0$ and the same maximal severity $s_1=5$. Similar non-monotonic behavior is
observed for other values of the maximal severities. Elementary estimates demonstrate that the
probability $p$ of obtaining these (or larger) deviations to below from the FOD for single injuries
($s_1=5$, $s_2=s_3=0$)  for {\em all} cases with $1\leq s_2 \leq 4$
simultaneously is less than $10^{-10}$. The number of cases used for these
estimates are given in Table~\ref{Table:9}. If the second severity coincides with the
maximal one, $s_2=s_1=5$ then the FOD is larger than for single traumas.

\begin{table}{\footnotesize\begin{center}
\caption{\label{Table:9}Number of cases for the maximal severity $s_1=5$ and various $s_2$
and $s_3$.}
\begin{tabular}{c|cccccc|}
\cline{2-7} & \multicolumn{6}{c|}{ $s_3$}  \\ \hline
 $s_2$ & 0 & 1 & 2 & 3 & 4 & 5   \\ \hline
 0 & 1,376 &&&&& \\ \hline
 1 & 276 & 101 &&&& \\ \hline
 2 & 302 & 163 & 332 &&&\\  \hline
 3 & 577 & 243 & 645 & 1,580 && \\ \hline
 4 & 349 & 140 & 203 & 2,653 & 2,301 & \\  \hline
 5 & 387 & 102 & 95 & 807 & 2,159 & 1,842\\  \hline
\end{tabular}
\end{center}}
\end{table}

It may be convenient to have formulas for estimation of FOD. This smoothed FOD (${\rm
sFOD}_{s_1}$) is found for $s_1=2, \ldots, 5$ as a linear combination of $s_{2,3}$ and
$s^2_{2,3}$ (\ref{sFOD}). For $s_1=1$ the simple formulas do not have much sense and we have
to use a refined model with the inclusion of age (Sec.~\ref{Sec:refine}) The number of
cases is not sufficient for good approximation for this extended model. For $s_1=6$ the
number of cases is not sufficient and we use three bins for trauma severities marked by
the values of the coarse-grained variable $\hat{s}$: $ 0\leq s_2 \leq 2$ ($\hat{s}_2=0$,
48 cases), $3\leq s_2 \leq 4$ ($\hat{s}_2=1$, 53 cases), and $5\leq s_2 \leq 6$
($\hat{s}_2=2$, 38 cases). ${\rm sFOD}_6$ is presented as a quadratic function of
$\hat{s}_2$.

\begin{equation}\label{sFOD}
\begin{split}
  {\rm sFOD}_2=&0.01910 + 0.02124 s_2 + 0.00037 s_3 \\ &-0.01054s_2^2 -0.00084 s_3^2;\\
 {\rm sFOD}_3=&0.02202 + 0.00256 s_2 -0.00238 s_3\\ &+ 0.00099 s_2^2 + 0.00101 s_3^2; \\
 {\rm sFOD}_4=&0.06571 - 0.02075 s_2 - 0.03116 s_3 \\ &+ 0.00706 s_2^2 + 0.01086 s_3^2; \\
 {\rm sFOD}_5=&0.35899 -0.13335 s_2 - 0.10879 s_3 \\ &+ 0.02963 s_2^2+  0.02748 s_3^2; \\
 {\rm sFOD}_6=&0.80297 -0.08750 \hat{s}_2 + 0.06102\hat{s}_2^2.
\end{split}
\end{equation}
All the coefficients are estimated using weighted least squa\-res method. The weight of the
severities combination $(s_1,s_2,s_3)$ is defined as the sum of weights of the
corresponding trauma cases.

\paragraph{Medical commentary}

The complete outcome dataset derived from this work allows all patients to be included in the analysis 
of the effect of combined injuries. The counter-intuitive results from this analysis (some combinations 
of injuries seem to have better outcomes than a single injury of the same severity) provides a fertile 
area for further work. It may be that the explanation is technical, within the way that the continuum 
of human tissue destruction from trauma is reduced to a simple 5 point scale. Each point on the scale 
is actually a band that covers a range of tissue damage. There might also be a true physiological 
explanation for the lower lethality of combined injuries, as each injury absorbs some of the force of 
impact. The same concept is used in Formula 1, where the cars are designed to break into pieces, 
with each piece absorbing some of the impact. In humans there is a well known concept that the 
face can act as a `crumple zone' and mitigate effect of force on the brain. The effect of injury 
combinations shown in Table 6 is a novel finding that requires further analysis.

\subsection{Example. Time after trauma, non-monotone and multimodal mortality coefficients}

In the early 1980s a hypothetical statement was
published that the deaths from trauma have a trimodal distribution with the following
peaks: immediate, early and late death \cite{BakerTrunkey1980,Lowe1983}. This concept was
clearly articulated in a popular review paper in Scientific American \cite{Trunkey1983}.
The motivation for this hypothesis is simple: Trunkey \cite{Trunkey1983} explains that the
 distribution of death is the sum of three peaks:  ``The  first  peak  ({\em `Immediate  deaths'}) 
corresponds  to  people  who  die  very  soon  after  an  injury;  the  deaths 
in  this  category  are  typically  caused  by  lacerations  of  the  brain,  the 
brain stem,  the  upper spinal  cord, the heart or one of the  major blood 
vessels.  The second  peak  ({\em `Early  deaths'})  corresponds  to  people  who 
die within the  first few  hours after an injury;  most of these  deaths are 
attributable  to  major  internal  hemorrhages  or  to  multiple  lesser 
 in­juries  resulting  in  severe  blood  loss.  The  third  peak  ({\em `Late  deaths'}) 
corresponds  to  people  who  die  days  or  weeks  after  an  injury;  these 
deaths  are  usually  due  to  infection  or  multiple  organ  failure.''

Strictly speaking, the {\em sum of three peaks does not have to be a trimodal distribution}. 
Many groups have published refutations of trimodality: they did not find the trimodal
distribution of death. In 1995, Sauaia et al reported that the ``greater proportion of late
deaths due to brain injury and lack of the classic trimodal distribution''
\cite{Sauaia1995}. Wyatt et al could not find this trimodal distribution
in data from the Lothian and Borders regions of Scotland between 1 February 1992 and 31
January 1994 \cite{Wyatt1995}. They hypothesised that this may be (partly) due to
improvements in care.

Recently, more data has become available and many such reports have been published
\cite{Demetriades2005,Knegt2008,Chalkley2010}. 
 {The suggestion that the improvement in care has led to the destruction of the second and third peaks 
has been advanced a number of times \cite{Knegt2008}. }In 2012, Clark et al performed an analysis of the distribution of survival
times after injury using interval censored survival models \cite{Clark2012}. They
considered the trimodal hypothesis of Trunkey as an artifact and provide arguments that
the observed (in some works) second peak is a result of differences in the definition of
death.

\begin{figure*}[!]
\begin{centering}
\includegraphics[width= 0.7\textwidth]{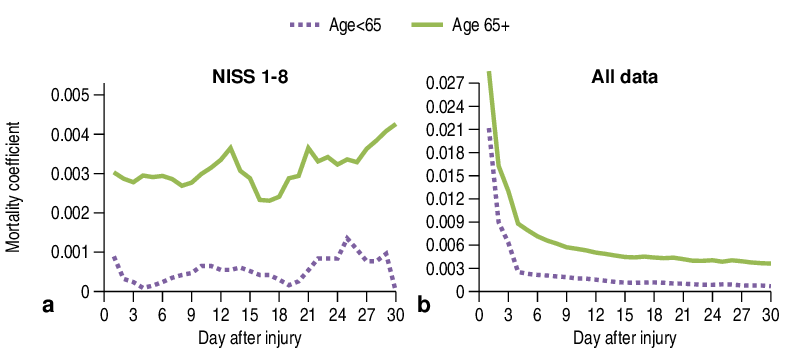}
 \caption{ \label{Fig:MortalityCoeff} Daily coefficient of
 mortality -- evaluated probability of a patient to die on day  $t$ under condition that he/she survived during days $1\div t-1$: a) for NISS=1-8, b) for all dataset. The coefficient is filtered by moving 5-day average starting from the 3rd day. The mortality coefficients are evaluated with the Markov models with retarded transfer. Data for age$<65$ and  {age$\geq 65$} are represented separately.}
\end{centering}
\end{figure*}

%{MortalityNISS.eps}
%{Mortality.eps}
%{MortalityMark.eps}

K. S{\o}reide et al analysed the time distribution from injury to death stratified by
cause of death. They demonstrated that the trimodal structure may be, probably, extracted
from data but its manifestation is model--dependent (see Fig.~6 in \cite{Soreide2007}).
There were several discussion papers published: ``Trimodal temporal distribution of  fatal trauma -- Fact
or fiction?'' \cite{Aldrian2008,Kruger2008}.

The trimodal hypothesis was tested on TARN data \cite{LeckieTARN}. It was demonstrated that
``the majority of in hospital trauma deaths occur soon after admission without further
peaks in mortality''. We reproduce the same results, indeed. But TARN database,  the
largest European trauma database, allows us to make a {\em stratified analysis of
mortality} and the preliminary results demonstrate the richness of the possible patterns of
death.

Let us test the famous Trunkey hypothesis. In Fig.~\ref{Fig:MortalityCoeff} the daily mortality coefficients are presented for low severities (a) (NISS severities 1-8, 27,987 cases in database, 508 death in TARN,  3,983 patients transferred from TARN within 30 days after injury), and for the whole database (b). For the prediction of death in the `OUT30' group we used
the model with retarded transfer.

The non-monotonicity and peaks in the mortality for low severities of injury are
illustrated in Fig.~\ref{Fig:MortalityCoeff}. Further analysis of these patterns
should involve other attributes such as the age of the patient and the type and localization of the
injury.

\paragraph{Medical commentary}

It has been widely accepted that the Trunkey trimodal distribution was a theoretical concept designed to illustrate the different modes of dying following injury. Previous analysis of trauma data has looked at all patients and has not shown any mortality peaks, however this new analysis shows that there are peaks (patterns) if subgroups are studied. The underlying clinical or patient factors are not immediately obvious, but future analysis giving a better understanding of patterns of death could act as a stimulus to look for the clinical correlates of these patterns - with the potential to find modifiable factors. The pattern of death in various subgroups as shown in Figure 7 is a novel finding that requires further analysis.

\section{Discussion}

Handling of data with missed outcomes is one of the first data cleaning tasks. For many
healthcare datasets, the problem of lost patients and missed outcomes (in 30 days, in six
months or any other period of interest) is important. There are two main approaches 
for solving this problem:
\begin{enumerate}
\item To find the lost patients in other national and international databases;
\item To recover the distribution of the missed outcomes and all their correlations
using statistical methods, data mining and stochastic modelling.
\end{enumerate}
Without any doubt the first approach is preferable if it is available: it is better to
have complete information when it is possible. Nevertheless, there may be various
organizational, economical and informational restrictions. It may be too costly to
find the necessary information, or this information may be unavailable or even does not
exist in databases. If there are only small number of lost cases (dozens or even
hundreds) then they may be sought individually. However if there are thousands of losses
then we need either a data integration system with links to appropriate databases
like the whole NHS and ONS data stores (with the assumption that the majority of the
missed data may be taken from these stores) or a system of models for the handling of missed
data, or both because we might not expect all missed data to be found in other databases.

In the TARN dataset, which we analyze in this paper, the outcome is unavailable for 19,289 patients.
The available case study paradigm cannot be applied to deal with missed outcomes because they
are not missed `completely at random'.  Non-stationary Markov models of missed outcomes
allow us to correct the fraction of death. Two na\"ive approaches give 7.20\% (available
case study) or 6.36\% (if we assume that all unknown outcomes are `alive'). The corrected
value is 6.78\% (refined model with retarded transfer). The difference between the
corrected and na\"ive models is significant, whereas the difference between different
Markov corrections is not significant despite the large dataset.

Non-stationary Markov models for unknown outcomes can utilize any scheme of predictive
models with using any set of available attributes. We demonstrate the construction of such models using 
maximal severity model, binned NISS model and binned NISS supplemented by the age
structure at low severities. We use weighting adjustment to compensate for the effect of
unknown outcomes. The large TARN dataset allows us to use this method without significant
damage to the statistical power.

Analysis of mortality for a combination of injuries gives an unexpected result. If
$s_1\geq s_2 \geq s_3$ are the three maximal severities of injury in a trauma case then
the expected mortality (FOD) is not a monotone function of $s_3$, $s_3$, under given $s_1$. For
example, for $s_1=4, 5$ expected FOD first decreases when $s_{2,3}$ grow from 0 to 1-2 and then
increases when $s_2$ approaches $s_1$.  {Probably more attributes, such as type of injury (blunt/penetrating), localization of traumas, gender, or airway status of the patient should be taken into account for further analysis to resolve this puzzle.}

Following the seminal Trunkey paper \cite{Trunkey1983}, multimodality of the mortality
curves is a widely discussed problem. For the complete TARN dataset the coefficient of
mortality monotonically decreases in time but stratified analysis of the mortality
gives a different result: for lower severities FOD is a
non-monotonic function of the time after injury and may have maxima at the second and
third weeks after injury. Perhaps, this effect may be (partially) related to geriatric
traumas.

It is important to stress that both effects, non-monotone  dependence of mortality on the severity vector of combined traumas and multimodality of the mortality curves for low severities, do not depend on the method of mortality correction. These effects manifest themselves for both na\''ive approaches as well as for Markov models.

We found that the age distribution of trauma cases is strongly multimodal
(Fig.~\ref{Fig:Mortality}). This is important for healthcare
    planning.

The next step should be the handling of missed values of input attributes in the TARN database
Firstly,  we should follow the ``Guidelines for reporting any analysis potentially
affected by missing data'' \cite{Sterne2009}, report the number of missing values for
each variable of interest, and try to ``clarify whether there are important differences
between individuals with complete and incomplete data''. Already preliminary analysis of the 
patterns in the distribution of the missed input data in the TARN
dataset demonstrates that the gaps in data are highly correlated and need further careful
analysis. Secondly, we have to test and compare various methods of handling missing input attributs in the TARN database.

It is not necessary to analyse all attributes in the database for mortality prediction and risk evaluation. 
It is demonstrated that there may exist an optimal set of input attributes for mortality prediction in 
emergency medicine and additional variables may even reduce the value of predictors \cite{Goodacre2006}. 
Therefore, before the analysis of imputation efficiency, it is necessary to select the set of most relevant variables of interest.

The models  developed in this case study can be generalized in several directions. Firstly,
for trauma datasets, different attributes could be included in the `state' $s$ for the
non-stationary Markov models (Figs.~\ref{Fig:MortalityModelBefore},
\ref{Fig:MainMarkovModelAfter}). We did not explore all such possibilities but have
studied just simple models of the maximal severity (Model~\ref{mod2}) and binned NISS (Model~\ref{mod3}). An example of
model refinement with inclusion of age in the state variable $s$  is presented in
Section~\ref{Sec:refine}. Secondly, the `two stage lottery' non-stationary Markov model
could be used as a general solution applicable to any health dataset where `TRANSFER
IN' or `TRANSFER OUT' is a feature. Transfer between hospitals is common in healthcare,
therefore, we expect that models of this type will be useful for all large healthcare
data repositories.

\section{Summary}

\begin{enumerate}
\item The Trauma Audit and Research Network (TARN) have collected the largest European trauma
    database. We have analysed 192,623 cases from the TARN database. We excluded from the analysis
    16,693 patients (8.67\%), who arrived into TARN hospitals later than 24
    hours after injury. The other 146,270 patients (75.94\%) approached TARN during the first
    day of injury and remained in TARN or discharged to a final destination within 30
    days of injury. 19,289 patients (13.19\%) from this group transferred from
    TARN to another hospital or institution (or unknown destination) within 30 days of 
    injury. For this subgroup the outcome is unknown.
\item Analysis of the missed outcomes demonstrated that they cannot be considered as
    misses `completely at random'. Therefore, the analysis of available cases is not applicable for the TARN database.
    Special efforts are needed to handle data with missed outcomes.
\item We have developed a system of non-stationary Mar\-kov models for the handling of missed
    outcomes  and validated these models on the data arising from patients who moved to TARN (and excluded from the
    model fitting). We have analysed  mortality in the TARN database using the Markov models which we have developed and also validated.
\item The results of analysis were used for weighting adjustment in the available cases
    database (reweighting of the death cases). The database with adjusted weights can be
    used for further data mining tasks and will keep the proper fraction of deaths.
\item The age distribution of trauma cases is essentially multimodal, which is 
    important for healthcare planning.
\item Our analysis of the mortality coefficient in the TARN database demonstrates that (i)
    for complex traumas the fraction of death is not a monotone function of all
    severities of injuries and (ii) for lower severities
    the fraction of death is not a monotonically decreasing function of time after
    injury and may have intermediate peaks in the second and third weeks after injury.
\item The approach developed here can be applied to various healthcare
    datasets which have the problem of lost patients, inter--hospitals transfer and missing outcomes.
\end{enumerate}

\section*{Acknowledgements}Supported by TARN and the University of Leicester.

\vspace{15mm}

{\bf Evgeny Mirkes} (Ph.D., Sc.D.) is a Research Fellow at the University of Leicester.
He worked for Russian Academy of Sciences, Siberian Branch, and Siberian Federal
University (Krasnoyarsk, Russia). His main research interests are biomathematics, data
mining and software engineering, neural networks and artificial intelligence. He led and
supervised many medium-sized projects in data analysis and development of
decision-support systems for computational diagnosis and treatment planning.

{\bf Timothy J. Coats} (FRCS (Eng), MD, FCEM) is a Professor of Emergency Medicine at the
University of Leicester. Chair FAEM Research Committee 2000-2009, Chair Trauma Audit and
Research Network (TARN), Chair NIHR Injuries and Emergencies National Specialist Group.
Research Interests: Diagnostics and monitoring in Emergency Care, Coagulation following
injury, Predictive modeling of outcome following injury.

{\bf Jeremy Levesley } (Ph.D, FIMA) is a Professor in the Department of Mathematics at the University of Leicester. His
research area is kernel based approximation methods in high dimensions, in Euclidean space and on manifolds.
He is interested in developing research at the interface of mathematics and medicine, and sees interpretation of 
medical data sets as a key future challenge for mathematics.

{\bf Alexander N. Gorban} (Ph.D., Sc.D., Professor) holds a Personal Chair in applied
mathematics at the University of Leicester since 2004. He worked for Russian Academy of
Sciences, Siberian Branch (Krasnoyarsk, Russia), and ETH Z\"urich (Switzerland), was a
visiting Professor and Research Scholar at Clay Mathematics Institute (Cambridge, MA),
IHES (Bures-sur-Yvette, \^{I}le de France), Cou\-rant Institute of Mathematical Sciences
(New York), and Isaac Newton Institute for Mathematical Sciences (Cambridge, UK). His
main research interests are dynamics of systems of physical, chemical and biological
kinetics; biomathematics; data mining and model reduction problems.

\end{document}